\documentclass[12pt]{iopart}

\usepackage{subfigure}

\usepackage{graphicx}

\usepackage{comment}
\usepackage{iopams}

\newcommand{\bea}{\begin{eqnarray}}
\newcommand{\eea}{\end{eqnarray}}
\newcommand{\beq}{\begin{equation}}
\newcommand{\eeq}{\end{equation}}
\newcommand{\nn}{\nonumber}
\newcommand{\eqref}[1]{(\ref{#1})}

\newcommand{\abs}[1]{\vert #1\vert}
\newcommand{\avg}[1]{\left\langle #1 \right\rangle}
\newcommand{\expect}[2]{\left\langle #1 \right\rangle_{#2}}

\newcommand{\half}{{\frac{1}{2}}}

\newcommand{\quarter}{{\frac{1}{4}}}

\newcommand{\dH}{{d_H}}
\newcommand{\dS}{{d_s}}
\newcommand{\cupper}{{\overline c}}

\newcommand{\clower}{{\underline c}}

\newcommand{\C}[1]{{\mathcal{#1}}}

\newcommand{\R}[1]{{\mathrm{#1}}}

\newcommand{\CGW}{{\mathrm{GW}}}
\newcommand{\Prob}{{\mathrm{Prob}}}

\newtheorem{theorem}{Theorem}
\newtheorem{lemma}[theorem]{Lemma}

\newtheorem{assume}[theorem]{Assumption}

\newenvironment{proof}[1][Proof]{\begin{trivlist}
\item[\hskip \labelsep {\bfseries #1}]}{\end{trivlist}}

\newcommand{\qed}{\nobreak \ifvmode \relax \else
      \ifdim\lastskip<1.5em \hskip-\lastskip
      \hskip1.5em plus0em minus0.5em \fi \nobreak
      \vrule height0.75em width0.5em depth0.25em\fi}
\providecommand{\href}[2]{#2}

\begin{document}

\title[Multigraph models for causal quantum gravity]{Multigraph models for causal quantum gravity and scale dependent  spectral dimension}
\author{Georgios Giasemidis, John F Wheater and Stefan Zohren}
\address{Rudolf Peierls Centre for Theoretical Physics, 1 Keble Road, Oxford OX1 3NP, UK}
\ead{ \mailto{g.giasemidis1@physics.ox.ac.uk}, \mailto{j.wheater1@physics.ox.ac.uk}, \mailto{zohren@physics.ox.ac.uk}}

\begin{abstract} 
We study random walks on ensembles of a specific class of random multigraphs which provide an ``effective graph ensemble'' for the causal dynamical triangulation (CDT) model of quantum gravity. In particular, we investigate the spectral dimension of the multigraph ensemble for recurrent as well as transient walks. We investigate the circumstances in which the spectral dimension and Hausdorff dimension are equal and show that this occurs when $\rho$, the exponent for anomalous behaviour of the resistance to infinity, is zero. The concept of scale dependent spectral dimension in these models is introduced. We apply this notion to a multigraph ensemble with a measure induced by a size biased critical Galton-Watson process which has a scale dependent spectral dimension of two at large scales and one at small scales. We conclude by discussing a specific model related to four dimensional CDT which has a spectral dimension of four at large scales and two at small scales.  \end{abstract}

\submitto{\JPA}

%

\maketitle
\makeatletter
\renewcommand\tableofcontents{\section*{\contentsname}\@starttoc{toc}}
\makeatother
\tableofcontents
\newpage
\section{Introduction}

Models of (discrete) random geometry have recently received considerable attention both in theoretical physics and in probability theory, often motivated by models of quantum gravity (see \cite{Ambjorn:1997} for an overview).  As we will discuss in more detail later the models of interest are built on ensembles of random graphs which, while grown according to local rules, have non-trivial long distance geometrical properties. A large part of the study of these systems aims to understand as comprehensively as possible the long distance fractal structure of the graphs in these ensembles. 

A relatively simple characterisation of the fractal structure of the geometries is their dimensionality. One possible measure of the dimensionality is the  \emph{Hausdorff dimension} $\dH$ which is defined  in terms of $B_N$, the volume of a ball of radius $N$ centred on a fixed point, by\footnote{We define the relationship
$ X_N\sim N^\rho$
to mean that  there exists a finite constant $N_0$ such that for $N>N_0$
\bea \clower N^{\rho} \left(\log N\right)^{\clower'}<X_N< \cupper N^{\rho} \left(\log N\right)^{\cupper'}\eea
where $0<\clower< \cupper$, $ \clower' < \cupper'$ are  constants.  The relationship $f(x)\sim x^\rho$ for small enough $x$ is defined similarly. Throughout this paper we will use the letter c and its decorations to denote constants whose values may change from line to line.}
\beq \label{Hausdorff}
B_N\sim N^{\dH} 
\eeq
for $N\to\infty$ assuming the limit exists. The ensemble average value of  $\dH$ can be determined from  $\avg{B_N}$ while in some ensembles the stronger statement that $\dH$ takes the same value for almost all graphs can be made.  An important example is the generic random tree (GRT) \cite{Aldous1998,Durhuus:2006vk} which has Hausdorff dimension  $\dH=2$ both on the average and almost surely, and characterises the  branched polymer phase of Euclidean quantum gravity \cite{AW,tjjw}.

Another  notion of dimensionality is the  \emph{spectral dimension} $\dS$. Given the probability $p(t)$ that a simple random walk returns to its origin after $t$ steps then $\dS$ is defined by
\beq
p(t)\sim t^{-\dS/2}, \quad t\to\infty
\eeq
assuming the limit exists.
An important example in the context of quantum gravity is again the GRT which has  spectral dimension  $\dS=4/3$ \cite{Durhuus:2006vk}. We see that $\dS<\dH$ for the GRT which in fact saturates the right hand side of the relation 
\beq \label{relationdsdh}
\dH\geq \dS \geq \frac{2\dH}{1+\dH}
\eeq
 derived under certain assumptions for fixed graphs \cite{coulhon}.

Models for which the spectral and Hausdorff dimensions are well understood analytically include the GRT,  random combs \cite{Durhuus:2005fq}, random brushes \cite{brush} and non-generic trees \cite{non-generic,Stefansson:2012aa}.
In the  mathematical literature  percolation clusters have been studied intensively, see for instance \cite{Barlow:2005aa}. In all these examples the random walk is recurrent. An example of a non-recurrent situation is encountered in the study of biased random walks on combs \cite{biasedcombs}. The situation is somewhat less satisfactory for models more directly related to quantum gravity. Two dimensional Euclidean quantum gravity (equivalently the planar random graph ensemble) has $\dH=4$ but $\dS$ is not known although numerical simulations indicate that it is close to 2 \cite{earlynumsim,numerics,scaling}  and it is known that the subset of graphs with vertex order strictly less than some finite constant is recurrent \cite{benjamini}.  The causal dynamical triangulation (CDT)  approach to quantum gravity  (see \cite{Ambjorn:1998} and \cite{Ambjorn:2006} for a review) is amenable to analytic investigation in two dimensions; the uniform infinite causal triangulation (UICT)  \cite{Durhuus:2009sm,Krikun,Sisko} -- essentially a CDT constrained never to die out -- has $\dH=2$ and it has been proved that $\dS\le 2$  \cite{Durhuus:2009sm}. Three and four dimensional CDT appear out of reach analytically but numerical simulations have produced some intriguing results. 
 In particular,
 the phenomenon of a scale dependent spectral dimension 
 has been observed in both three  \cite{benedetti} and four \cite{Ambjorn:2005db} dimensions. 
 These numerical results indicate that in the continuum limit there is a typical length scale of the model (in this case the Planck length) and that for lengths of the random walk much lower than this scale the spectral dimension is $d_s^0=2$ while for walk lengths much greater than this scale the spectral dimension is $d_s^\infty=4$. Such behaviour has also been observed in other approaches of quantum gravity, most importantly the exact renormalization group approach \cite{Litim,Lauscher,Reuter:2011ah} and Ho\v rava Lifshitz gravity \cite{Horava:2009aa,Horava:2009ab}, as well as in models of spin foams \cite{Modesto:2008jz,Caravelli:2009gk,Magliaro:2009if}; very recently there has been considerable interest in the proposal that in three and four dimensions CDT and Ho\v rava Lifshitz gravity are equivalent \cite{HLrel,Anderson:2011bj}. The circumstances under which the lhs of \eqref{relationdsdh} is saturated has been investigated very recently in a continuum formulation \cite{Calcagni:2011kn,Calcagni:2012aa}.

The numerical results for four dimensional CDT raise two immediate questions. Firstly whether there is a rigorous definition of scale dependent spectral dimension in the context of graph ensembles and secondly whether there might be a reduced version of the full CDT model which is analytically tractable and yet displays behaviour similar to the full model, at least insofar as the spectral dimension is concerned.

The first question was answered in the affirmative by \cite{Atkin:2011ak} where it was shown that for a graph ensemble containing a length scale in the measure there is a consistent definition of separate spectral dimensions $d_s^0$ and $d_s^\infty$  for walks respectively  shorter  and longer than the scale. Furthermore it was shown that there exist toy models of random combs with a length scale in their measure which exhibit this property;  those models should be regarded as  ``kinematical'' toy models because they do not have a completely local growth rule and are therefore not directly related to any model of quantum gravity. The construction is reviewed briefly in Section \ref{sec3} of this paper.

\begin{figure}
\begin{center}
\includegraphics[scale=0.5]{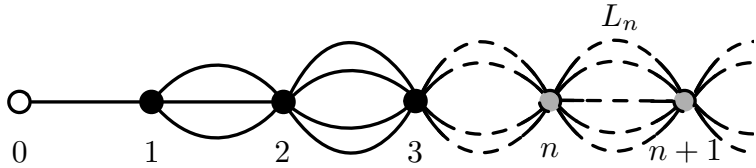}
\caption{An example of a multigraph.}
\label{multigraph_example}
\end{center}
\end{figure}

To address the second question in this paper we will analyse the spectral dimension of a variety of  \emph{multigraph ensembles} which we will define carefully in Section \ref{sec2}. An example of the graphs we consider is shown in Figure \ref{multigraph_example}; essentially they are constructed from the discretized half line by adding extra edges between adjacent vertices with some probability distribution (not necessarily i.i.d.) for the final number of edges $L_n$ between vertices $n$ and $n+1$.
%
In the context of %
the CDT approach to quantum gravity %
multigraphs  were first introduced in \cite{Durhuus:2009sm} as part of the proof that $\dS\le 2$.
Figure \ref{fig2a} gives an example of a causal triangulation. 
Given an infinite causal triangulation we can obtain a multigraph by the  injective map in which we coalesce all vertices at a fixed distance $n$ from the root vertex 
into a single vertex, which becomes the vertex $n$ of the multigraph, while retaining only the edges whose two ends are at different distances from the root; the measure induced on the $\{L_n, n=0,1,\ldots\}$ is then inherited from the measure on the causal triangulation.  It can be shown that the spectral dimension of the multigraph is an upper bound on that of the causal triangulation which is believed to be tight.

Each  two dimensional causal triangulation is in bijection with a tree (see Figure \ref{fig2b}) and  the uniform measure on infinite causal triangulations is equivalent to the uniform measure on the GRT  \cite{Durhuus:2006vk} -- equivalently the measure of a critical Galton-Watson process conditioned on non-extinction \cite{Aldous1998,Athreya}.  Thus the measure on the $\{L_n\}$ is known 
and the spectral dimension of this multigraph ensemble was shown in \cite{Durhuus:2009sm} to be $\dS=2$ almost surely. Further, defining $B_N=\sum_{n=0}^N L_n$ one also has $\dH=2$ almost surely  thus  saturating the left hand side of \eqref{relationdsdh}. While this multigraph ensemble only gives an upper bound for the spectral dimension of the UICT, it is believed that both have the same spectral dimension.

The map from triangulation to multigraph applies to triangulations of any dimension although in dimensions greater than two  we do not know the resulting measure for $\{L_n\}$. In this paper we investigate the use of multigraph ensembles to provide an ``effective graph ensemble" description of at least some aspects of the full CDT ensembles, and combine it with the methods of \cite{Atkin:2011ak} to study scale dependent spectral dimension. In particular, we will give an analytic description in terms of this effective ensemble of the phenomenon of  scale dependent spectral dimension which has been observed in computer simulations of four-dimensional CDT \cite{Ambjorn:2005db} and find a concrete model with $d_s^0=2$ and $d_s^\infty=4$.

\begin{figure}[t]
\centering
\subfigure[A causal triangulation]{
\includegraphics[width=6cm]{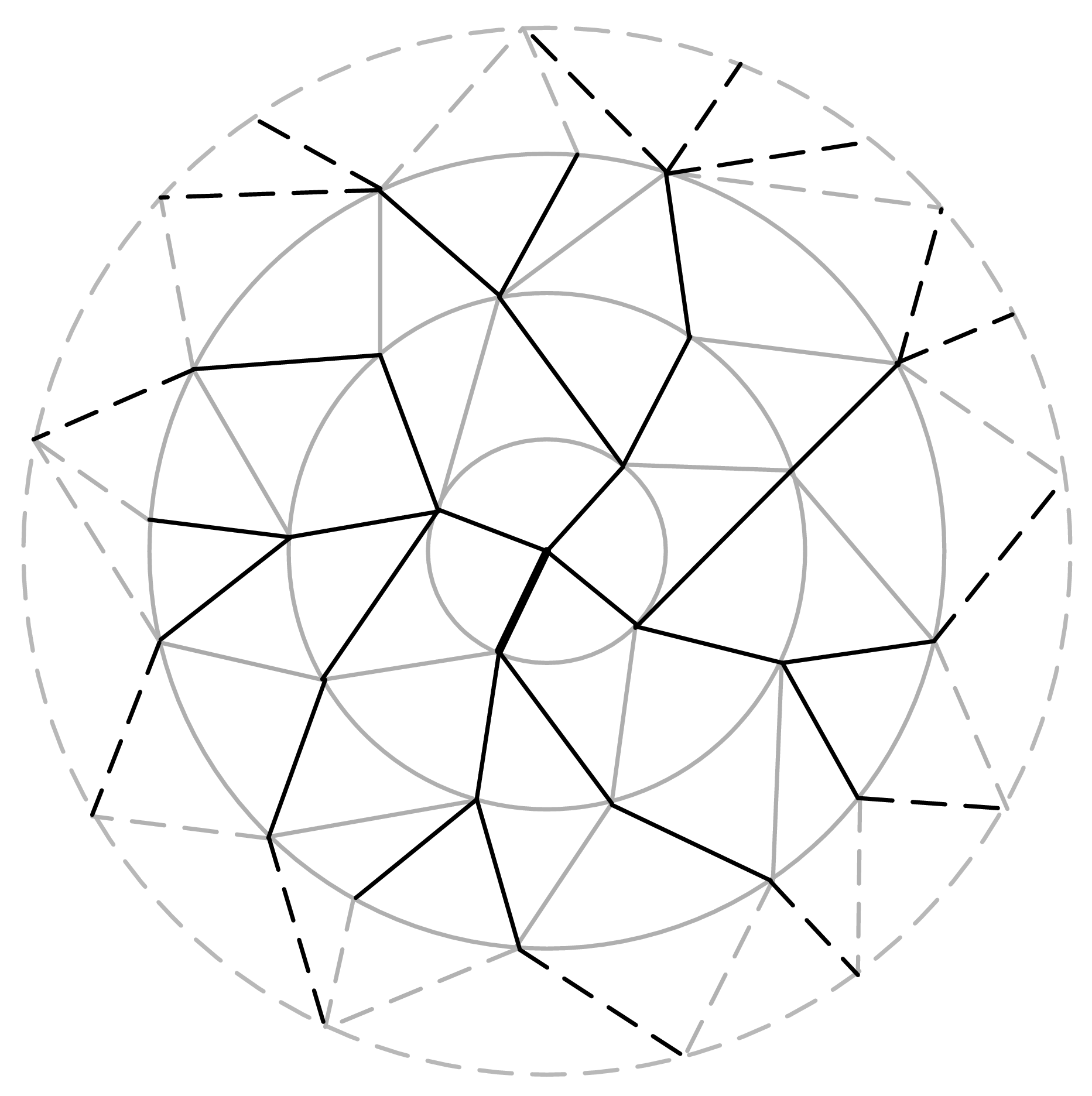}
\label{fig2a}
}
\subfigure[Tree bijection]{
\includegraphics[width=6cm]{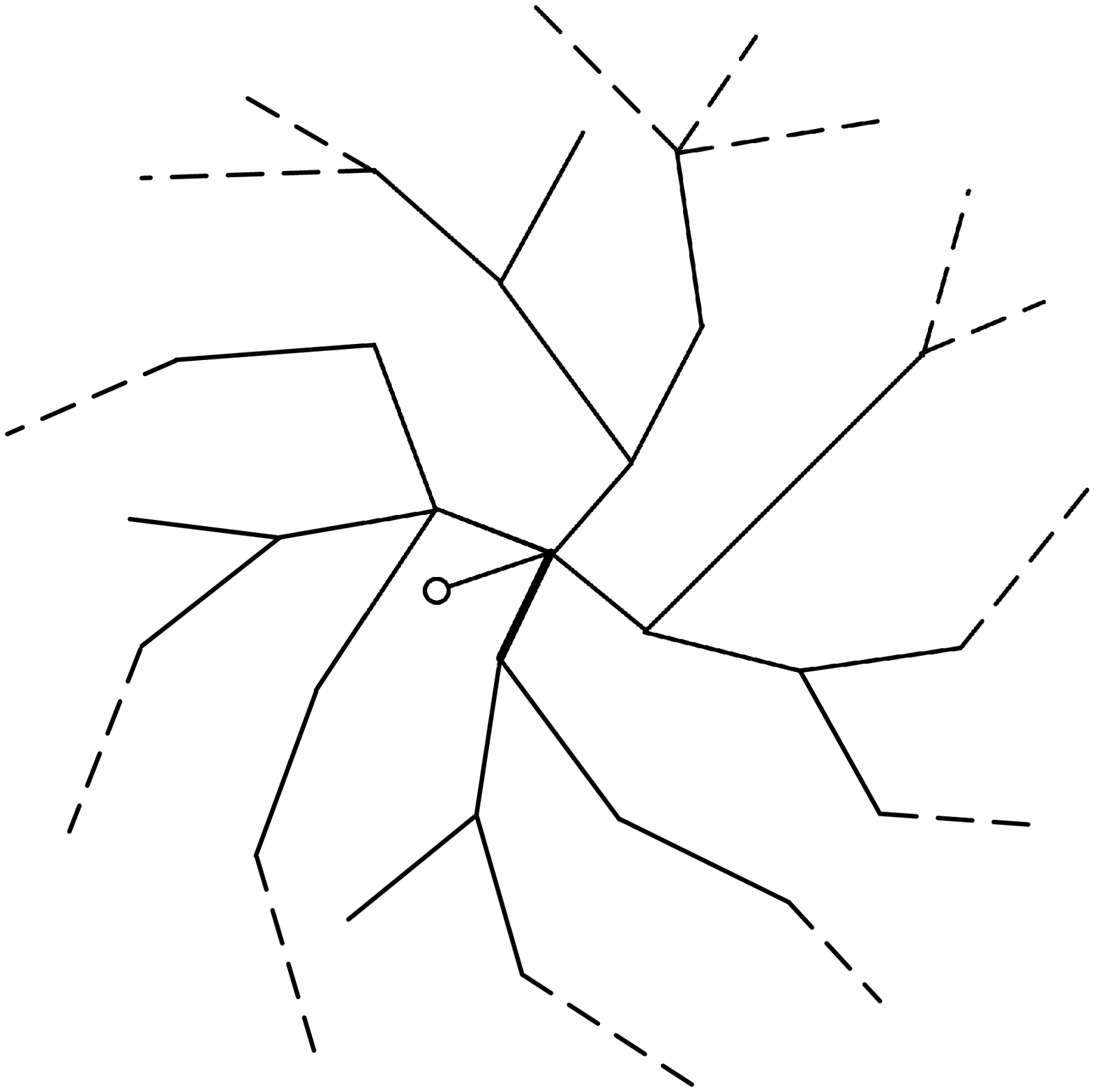}
\label{fig2b}
}
\caption{From causal triangulation to trees. The bijection works as follows using the definitions of Section \ref{sec2}.1; delete the rightmost edge of each vertex from $S_k$ to $S_{k+1}$ (time-like edges) and all edges connecting vertices at height $k$ (space-like edges) - deleted edges are drawn as grey lines. Add a vertex $r$ (empty circle) and mark an edge from $S_0$ to $S_1$ which is the rightmost edge with respect to the edge $(r, S_0)$ (fat line). Dashed lines encode the fact that both graphs are infinite and extend beyond finite height.}
\end{figure}

This paper is organized as  follows. After an introduction to multigraphs and the measure induced by Galton-Watson processes in Section \ref{sec2} and basic definitions regarding random walks and (scale dependent) spectral dimension in Section \ref{sec3}, we present our main results in the remaining sections: \\
\noindent\emph{Scale dependent spectral dimension in the recurrent case --} In Section \ref{sec4} we describe a multigraph ensemble based on a  critical Galton-Watson process  conditioned  never to die out
which has spectral dimension $d_s^0=1$ at short distances while at long distances $d_s^\infty=2$.\\
\noindent\emph{Properties of spectral dimension in the non-recurrent case --} In Section \ref{sec5} we discuss the general properties of transient multigraph ensembles restricting ourselves to the regime of physical interest with $2\leq \dH \leq 4$ and establish in particular the relationship
\beq
\dS=\frac{2\dH}{2+\rho},
\eeq 
where  $\rho$ is the resistance exponent.\\%
\noindent\emph{Scale dependent spectral dimension in the non-recurrent case --} In Section \ref{sec6} we show that a multigraph ensemble with certain geometric properties motivated by the concrete model of Section \ref{sec4}  would have spectral dimension $d_s^0=2$ at distances much less than the Planck length while at  distances much greater than the Planck length $d_s^\infty=4$.

\section{Multigraphs and measures}\label{sec2}

In this section we define the various sets of graphs and probability measures that we will need later in the paper. For the reader's convenience we adhere closely to the definitions and notation of \cite{Durhuus:2006vk, Durhuus:2009sm} which may be referred to for further details.

\subsection{Multigraphs and trees}

Throughout this paper we will be concerned with the set ${\C R}$  %
 constructed  from the 
non-negative integers regarded as a graph so that $n$ has neighbours $n\pm 1$, 
except for 0, referred to as the root $r$, which only has 1 as a neighbour, and so that there 
are
$L_n(G)\ge1$ edges 
connecting $n$ and $n+1$.  A multigraph $G\in \C R$ is completely described by listing the number of edges $\{L_n(G), n=0,1,\ldots\}$; an example is shown
in Figure \ref{multigraph_example}.  A multigraph ensemble $\R M=\{\C R,\chi\}$ consists of the set of graphs $\C R$ together with a probability measure $\chi(\{G\in \C R: \C A\}) $ for the (finite) event $\C A$. 

A rooted tree $T$ is a connected planar graph consisting of vertices $v$ of finite degree connected by edges but 
containing no loops; the root $r$ is a  marked vertex connected to only one edge. Denoting the number of edges in a tree by $\abs T$, the 
set of all trees $\C T$ contains the set of finite trees ${\C T}_f =\bigcup_{N\in {\mathbb N}} {\C T}_N$ 
where ${\C T}_N=\{ T\in  {\C T} : {\abs T} =N\} $ 
and the set of infinite trees $\C T_\infty$. A tree ensemble $\R T=\{\C T,\mu\}$ consists of the set of graphs $\C T$ together with a probability measure $\mu(\{G\in \C T: \C B\}) $ for the (finite) event $\C B$.

Define the height $h(v)$ of a  vertex $v$ to be the graph distance from $r$ to $v$,  $S_k(G)$ to be the set of all vertices of  $G$ having height $k$,  $D_k(T)$ to be the set of edges in a tree $T$ connecting vertices at height $k-1$ to vertices at height $k$, and $B_k(G)$ to be the number of edges of $G$ in the ball of radius $k$ centred on the root  (so $h(r)=0$ in all cases, and for example $\abs{S_r(G\in \C R)}=1$). We then define a mapping $\gamma: \C T_\infty\to\C R$  which acts on a tree $T$ by  identifying all vertices $v\in S_k(T)$ with the single vertex $u$ but retaining all the edges. The resulting multigraph ensemble inherits its measure from the measure on $\C T_\infty$
so that  for  integers
$0\le k_1<\dots<k_m$ and positive integers $M_1,\dots, M_m$ 
\bea \lefteqn{\chi(\{G\in \C R: L_{k_i}(G)=M_i,\, i=1\ldots
  m\})}
  \nn\\
&=&\mu(\{T\in\C T_\infty:\abs{D_{k_i+1}(T)}=M_i,\, i=1\ldots m\}).\label{RRequivGRT}\eea
This mapping is well defined provided that the measure on $\C T_\infty$ is supported on trees with a unique path to infinity.  It is convenient to use trees to define some of  the ensembles of multigraphs we will be considering because we can thereby exploit many standard results.

\subsection{Branching Processes and Trees}

In a Galton Watson (GW) process each member of a given generation has $k$ offspring with probability $p_k$ where $p_0> 0$ and $p_i>0$ for at least one $i\ge 2$.
It is convenient to introduce the generating function 
\beq f(x)=\sum_{n=0}^\infty   p_n x^n 
\eeq
which satisfies $f(1)=1$. 
The process is called \emph{critical} if $f'(1)=1$ and \emph{generic} if $f(x)$ is analytic in a neighbourhood of the unit disk.

A GW process can be associated with a tree by identifying each member with a vertex and drawing edges between members and their offspring. The founder member is attached by an edge to a special vertex called the root. 
The probability distribution $\mu_\CGW$ on finite trees $\C T_f$ is then
\beq 
\mu_\CGW(T ) = \prod_{v\in T\setminus r} p_{\sigma_v-1},\label{treemeasure}
\eeq
where $\sigma_v$ is the degree of vertex $v$, and the ensemble $(\C T_f,\mu_\CGW)$ is called a critical Galton Watson tree.

The probability distribution $\mu_N$ on trees of a  fixed size $\C T_N$ is given by
\beq \label{muNdef}
\mu_N(T)=Z_N^{-1}\prod_{v\in T\setminus r} p_{\sigma_v-1}
\eeq
where
\beq 
Z_N=\sum_{T\in\C T_N}\prod_{v\in T\setminus r} p_{\sigma_v-1}.
\eeq
The {\it single spine trees} are the  subset $\C S$ of the infinite trees whose members consist of a 
single infinite linear chain, called the \emph{spine}, to each vertex of which are attached a 
finite number of finite trees by identifying their root with that vertex. 
The following result was established in \cite{Durhuus:2006vk}. 
\begin{theorem}\label{thm1}
Assume that $\mu_N$ is defined as above as a probability measure on $\C T$ where $\{ p_n\}$ defines a
generic and critical GW process.  Then
\beq
 \mu_N\to\mu_\infty\quad as \quad N\to\infty\label{mudef}
 \eeq
where $\mu_\infty$ is a probability measure on $\C T$ concentrated on the set of single spine trees $\C S$.
The generating function for the probabilities for the number of finite
branches at a vertex on the spine is $f'(x)$. Moreover, the individual
branches are independently and identically distributed according to the original critical
GW process. 
\end{theorem}
We call an ensemble  $(\C S , \mu_\infty)$   a generic random tree (GRT).  Denoting expectation values in the measure $\mu_X$ by $\expect{ \cdot }{X}$ some useful results on the attributes of the tree ensembles are summarized by

\begin{lemma}\label{StandardResults}%
For any critical Galton Watson ensemble and related  generic random tree ensemble,
\bea 
\expect {\abs{D_k}}{\infty} &=& (k-1)f''(1)+1,\quad k\ge 1,\label{Done}\\
\expect{B_k}{\CGW}&=&k,\quad k\ge 1,\label{BGW}\\
\expect{B_k}{\infty}&=&\half k(k-1)f''(1)+k,\quad k\ge 1,\label{BT}\\
 \expect {\abs{D_k}^{-1}}{\infty} &=& \mu_{\CGW}(  D_k(T)> 0),\quad k\ge 1.\label{Dinv}\eea
\end{lemma}
\begin{proof}The proofs of \eqref{Done},  \eqref{BGW} and \eqref{BT}  are given in \cite{Durhuus:2006vk}, Appendix 2. The result \eqref{Dinv} uses Lemma 4 and  the proof of  Lemma 5 of \cite{Durhuus:2006vk}. \end{proof}

We will find it particularly useful to consider  the uniform process $U$ for which%
\bea  p^U_k&=&\cases{ b,&$k=0,$\\      b^{k-1}(1-b)^2,& $k\ge 1$,}\label{Udist}\eea
with $0<b<1$. The generating function is
\bea f^U(x)=\sum_{k=0}^\infty p^U_k x^k=\frac{b+(1-2b)x}{1-bx}.\eea
The $r$'th iterate of $f^U$ is 
 \bea f^U_r(x)=\frac{rb - (rb + b - 1)x}{1 - b + rb - rbx} \label{Riterate}\eea
and 
\beq {f^U}''(1)=\frac{2b}{1-b}.\eeq
We will denote by $\infty U$ the GRT measure associated with $U$ (equivalently this is the Galton Watson process described by $U$ and constrained never to die out).
It follows from Lemma \ref{StandardResults} and \eqref{Riterate} that 
\bea \expect {\abs{D_k}^{-1}}{\infty U} &=& \frac{1}{1+  (k-1)b(1-b)^{-1}} =\frac{1}{1+{f^U}''(1)(k-1)/2},\,\, k\ge 1.\label{U:Dinv}\eea

\section{Random walk and spectral dimension}\label{sec3}

In this section we %
review the definition of spectral dimension and the nature of its scale dependence.  For the reader's convenience we adhere closely to the definitions and notation of \cite{Durhuus:2006vk, Durhuus:2009sm, Atkin:2011ak} which may be referred to for further details.

\subsection{Random walk}

We define a random walk on $G\in \C R$ by considering integer time steps labelled by $t$. A walker arriving at vertex $n\ne 0$ at time $t$ moves at time $t+1$ to vertex $n+1$ with probability $p_n(G)=L_n(G)(L_n(G)+L_{n-1}(G))^{-1}$ and to vertex $n-1$ with probability  $1-p_n$;
if the walk is at the root at time $t$ it moves to the vertex 1 with probability 1.  We denote by $\Omega_G$ the set of all walks on $G$, and by $\omega(t)$ the location of the walk $\omega\in\Omega_G$ at time $t$.

The generating function for the probability of first return to the root is given by
\bea\fl{ P_G(x)=\sum_{t,\omega\in\Omega_G} \Prob\left(\{\omega(t)=0\,\vert\, \omega(0)=0, \omega(t')>0, 0<t'<t\}\right)\,(1-x)^{\half t}}\eea
and for all returns to the root by
\beq Q_G(x)=\sum_{t,\omega\in\Omega_G} \Prob\left(\{\omega(t)=0\, \vert \, \omega(0)=0\right\})\,(1-x)^{\half t}.\eeq
The two generating functions are related by
\beq Q_G(x) = \frac{1}{1-P_G(x)}.\label{QPreln}\eeq
For two examples where one can determine an exact analytic expression for the generating function of the return probability see Appendix \ref{Solvable}. 

Now let the graph $G_n$ be obtained from $G$ by amputating the vertices $0,\ldots, n-1$ and all the edges attached to them, identifying the vertex $n$ of $G$ as the root of $G_n$ and relabelling the rest of the vertices in the obvious way.  Then because of the chain structure we have
\beq P_{G_n}(x)=\frac{(1-x)(1-p_{n+1}(G))}{1-p_{n+1}(G)P_{G_{n+1}}(x)}.\label{Precurr}\eeq
This leads  to the following useful result.
\begin{lemma}\label{monotonicity} {\bf (Monotonicity)}
For any $G\in\C R$ and $G'\in\C R$ which are identical except that  $L_k({G'})=L_k({G})-1$ for some $k>0$,
\beq P_{G}(x) < P_{G'}(x).\eeq
\end{lemma}
\begin{proof} First note that from \eqref{Precurr} $P_{G_{k-1}}(x)$ is a monotonically increasing function of $P_{G_{k}}(x)$ and therefore $P_G(x)$ is a monotonically increasing function of $P_{G_{k}}(x)$; it then suffices to prove the $k=1$ case. This is easily done by applying \eqref{Precurr} twice in succession to express both $P_G(x)$ and $P_{G'}(x)$ in terms of $P_{G_2}(x)$, $L_0(G)$, $L_1(G)$  and  computing the difference. Note that the lemma is not true for the case $k=0$.
\end{proof}
An immediate consequence is that
\bea \eta_G(x)&=&\frac{Q_G(x)}{L_0(G)}\label{eta:def}\\
&<&\frac{1}{ x^{1/2}}.\label{eta:upper}\eea
To obtain this we reduce all $L_{n>0}$ to 1, repeatedly applying Lemma \ref{monotonicity} to obtain the upper bound 
\beq P_G(x)<P_{G^*}(x)\eeq
where $G^*$ is the graph with edge numbers  $\{L_0(G),1,1,1,\ldots\}$ and then compute  $P_{G^*}(x)$ explicitly by the methods of \cite{Durhuus:2005fq}. 
 Since by definition $Q_G(x)\ge 1$ we also have 
that
\beq  \eta_G(x)\ge\frac{1}{L_0(G)}.\label{eta:lower}\eeq

\subsection{Spectral dimension and scale dependence}

The graph spectral dimension controls the return probability of long walks on the graph and is defined to be
\beq d_s=-2 \lim_{t\to\infty}  \frac{\log\left(\Prob\left(\{\omega(t)=0\, \vert \, \omega(0)=0\right\})\right)}{\log t}\label{tdependence}\eeq
if the limit exists.  In this paper we work with generating functions and define  the  graph spectral dimension to be $d_s$ if, as $x\to 0$, 
\bea Q(x)	&\sim& x^{-1+d_s/2},\quad d_s< 2,\nn\\
	\abs{Q'(x)}	&\sim& x^{-2+d_s/2},\quad 2\le d_s< 4.\label{xdependence}\eea
In the case $d_s=2$ we expect  that, as $x\to 0$, $Q(x)$ is either logarithmically divergent if the graph is recurrent or finite if the graph is transient.
By a Tauberian theorem the definitions \eqref{tdependence} and \eqref{xdependence} are equivalent 
provided the limits exist.

The question of scale dependent  spectral dimension  was discussed in detail in  \cite{Atkin:2011ak} where it was shown that models with this property exist. Clearly if we just  consider short random walks  on a graph we only see the local discrete structure and cannot expect  any scaling behaviour. Instead we construct a model in which  limits are  taken in such a way that all walks are long in graph units but the measure $\mu$ depends on  a parameter $\Lambda$  which is continuously variable and sets a distance scale relative to which walks can be either short or long.  
%
The following result was proved in \cite{Atkin:2011ak}
\begin{lemma}\label{Lemma:scaling}
Assume that there exist constants $\Delta_\mu$ and $\Delta$ such that
\beq \label{scalingQ}
\tilde Q(\xi,\lambda)=\lim_{a\to 0} a^{\Delta_\mu} \expect{Q\left(x=\xi a,\Lambda=a^{-\Delta}\lambda^\Delta\right)}\mu\label{tildeQ:definition}\eeq
exists and is non-zero and the combination $\xi\lambda$ is dimensionless. Then \\
i) there exists a $\tau_0(\xi,\lambda)$ such that 
\bea\fl \tilde Q(\xi,\lambda)(1-e^{-\xi\lambda})< \lim_{a\to 0} a^{\Delta_\mu}\sum_{t=0}^{\lfloor\tau_0/a\rfloor}
\Prob\left(\{\omega(t)=n\, \vert \, \omega(0)=n\right\})\,(1-\xi a)^{\half t} <  \tilde Q(\xi,\lambda)\eea
and $\tau_0(\infty,\lambda)=\lambda$;\\
ii) there exists a $\tau_1(\xi,\lambda)$ such that 
\bea\fl \tilde Q(\xi,\lambda)-e\lambda^\half< \lim_{a\to 0} a^{\Delta_\mu}\sum_{t=\lceil\tau_1/a\rceil}^{\infty}
\Prob\left(\{\omega(t)=n\, \vert \, \omega(0)=n\right\})\,(1-\xi a)^{\half t} <  \tilde Q(\xi,\lambda)\eea
and $\tau_1(0,\lambda)=\lambda$.
\end{lemma}
So as $\xi\to\infty$ we see that $ \tilde Q(\xi,\lambda)$ describes walks of continuum duration less than $\lambda$ and that as $\xi\to 0$ $ \tilde Q(\xi,\lambda)$ describes walks of continuum duration greater than $\lambda$ (provided it diverges in that limit). The spectral dimensions $d_s^\infty$ and $d_s^0$ in the long and short walk limits respectively are then defined by
\bea d^\infty_s&=& 2\left(1+\lim_{\xi\to 0} \frac{\log(\tilde Q(\xi;\lambda))}{\log \xi}\right),\nn\\
d^0_s&=& 2\left(1+\lim_{\xi\to\infty} \frac{\log(\tilde Q(\xi;\lambda))}{\log \xi}\right),\label{dsdef}
\eea
provided these limits exist (which is not at all assured but it was demonstrated in \cite{Atkin:2011ak} that there are models in which all the required limits do exist). 
Lemma \ref{Lemma:scaling} and \eqref{dsdef} yield the spectral dimension if the graphs are recurrent. It is straightforward to check that for transient graphs with $\dS\le 4$ an exactly analogous result relates $ \partial_\xi \tilde Q(\xi,\lambda)$ and $ \partial_x  Q(x,\Lambda)$ and that 
\bea d^\infty_s&=& 2\left(2+\lim_{\xi\to 0} \frac{\log|\partial_\xi\tilde Q(\xi;\lambda)|}{\log \xi}\right),\nn\\
d^0_s&=& 2\left(2+\lim_{\xi\to\infty} \frac{\log|\partial_\xi\tilde Q(\xi;\lambda)|}{\log \xi}\right).\label{trdsdef}
\eea

\section{Scale dependent spectral dimension in the recurrent case}  \label{sec4}

The model we consider in this section is a multigraph ensemble whose measure is related to the uniform GRT measure $\infty U$ through \eqref{RRequivGRT} and whose graph distance scale is therefore set by the parameter $b$. The weighting of the multigraphs in this ensemble can be related to an  action for the corresponding  CDT ensemble which contains a coupling to the absolute value of the scalar curvature. To see this first note that the probability for a finite tree $T$ in the $U$ ensemble is given by
\bea b^{\sum_{v\in T}\abs{\sigma_v-2}}\,(1-b)^{2\sum_{v\in T}1-\delta_{\sigma_v,1}}.\eea
The quantity $\abs{\sigma_v-2}$ is in fact the one dimensional analogue of the absolute value of the scalar curvature; small $b$ suppresses all values of vertex degree except $\sigma_v=2$. The causal triangulation $C$ which is in bijection to this tree  has probability
\bea b^{\sum_{v\in C}\abs{\sigma_{f(v)}-2}}\,(1-b)^{2\sum_{v\in C}1-\delta_{\sigma_{f(v)},1}},\eea
where $\sigma_{f(v)}$ is the number of  `forward' edges connecting vertex $v$ to vertices whose distance from the root is one greater than that of $v$.  Taking into account the causal constraint the effect of small $b$ on the triangulation is to suppress all vertex degrees except $\sigma_v=6$ so that $-\log b$ plays the role of a coupling to a term $\sum_v\abs{R_v}$, essentially the integral of the absolute value of the scalar curvature, in the action for the CDT.

The main result of  this section is
\begin{theorem}\label{VariableSpectralDimensionRecurrent}
The scaling limit of the multigraph ensemble  with measure $\infty U$
has spectral dimension $d_s^0=1$ at short distances and  $d_s^\infty=2$ at long distances.

\end{theorem}

\begin{proof}  Taking $\Lambda=b^{-1}$ and $\Delta_\mu=\Delta=\half$ in \eqref{tildeQ:definition} we obtain
   \bea
\tilde Q(\xi,\lambda)&\sim&\cases{ \xi^{-\half},&$\xi\gg\lambda^{-1},$\\ \lambda^\half\abs{\log \lambda\xi},& $\xi \ll\lambda^{-1},$\label{newthm1}}
\eea
the proof of which follows. The theorem then follows from Lemma \ref{Lemma:scaling} and the definitions \eqref{dsdef}.
\end{proof}

\subsection{Short distance behaviour: $\xi\to\infty$}

By the monotonicity lemma we have 
\beq Q(x,b)<x^{-\half}\eeq
from which it follows immediately that
\beq \tilde Q(\xi,\lambda) < \xi^{-\half}.\eeq
We can obtain a lower bound for the expectation value by considering only the contribution of graphs for which the first $N$ vertices have only one edge and  walks which get no further than $N$; then using \eqref{PhalflineLf} we get
\beq \expect{Q(x,b)}{\infty U}>(1-b)^{2N}x^{-\half}\frac{(1+x^\half)^N-(1-x^\half)^N}{(1+x^\half)^N+(1-x^\half)^N}.\eeq
Setting $N=b^{-1}$ we find 
\beq  \tilde Q(\xi,\lambda) > \xi^{-\half} e^{-2}\tanh(\sqrt{\xi\lambda})\eeq
which establishes  the first part of \eqref{newthm1}.%

\subsection{Lower Bound as $\xi\to 0$}

From now on to improve legibility  we will adopt the following simplified notation whenever it does not lead to ambiguity; for $\eta_{G_n}(x)$ we will write $\eta_n(x)$ and for $L_n(G)$ we will write $L_n$. We will also suppress the second argument $b$ in $\eta_n(x)$.

From \eqref{QPreln} and \eqref{Precurr} we obtain
\beq \eta_n(x)=\frac{\eta_{n+1}(x)+\frac{1}{L_n}}{1+xL_n\eta_{n+1}(x)}\label{start}\eeq
which can be iterated
to give 
\bea \eta_n(x)&=&\frac{\eta_{N}(x)}{ \prod_{k=n}^{N-1}(1+xL_k\eta_{k+1}(x))}   +\sum_{k=n}^{N-1}\frac{1}{L_k}\frac{1}{\prod_{m=n}^{k}(1+xL_m\eta_{m+1}(x))}\nn\\
		&>&\sum_{k=n}^{N}\frac{1}{L_k}\exp(-\sum_{m=n}^{k}xL_m\eta_{m+1}(x)),\label{EtaLower}
\eea
where we have used $\eta_N(x)\ge 1/L_N$. Using the monotonicity lemma \eqref{eta:upper}  and 
Jensen's inequality  gives
\beq \expect{\eta_0(x)}{\infty U}>
\sum_{n=0}^{N}\frac{1}{\expect{L_n}{\infty U}}\prod_{k=0}^{n}e^{-\sqrt{x}\expect{L_k}{\infty U}}.\eeq
Using Lemma \ref{StandardResults} and \eqref{U:Dinv} gives 
\bea \expect{\eta_0(x)}{\infty U}&>&
\sum_{n=0}^{N}\frac{1}{1+{f^U}''(1)n}
e^{-\sqrt{x}\sum_{k=0}^{n}1+{f^U}''(1)k}\\
&=&\sum_{n=0}^{N}\frac{1}{1+{f^U}''(1)n}
e^{-\sqrt{x}(n+1)(1+\half {f^U}''(1)n)}\nn\\
&>&e^{-\sqrt{x}(N+1)(1+\half {f^U}''(1)N)}  \sum_{n=0}^{N}\frac{1}{1+{f^U}''(1)n}\nn\\
&>&\frac{1}{{f^U}''(1)} e^{-\sqrt{x}(N+1)(1+\half {f^U}''(1)N)}\log \left(1+{f^U}''(1)N\right).
\eea
Now let $N=\lfloor b^{-\half} x^{-\quarter}\rfloor$ and set  $b=a^\half\lambda^{-\half}$, $x=a\xi$ so
\bea \tilde Q(\xi,\lambda)&=&\lim_{a\to 0}a^\half \expect{\eta_0(a\xi)}{\infty U} \\
&>&\frac{1}{2}{\lambda^\half}e^{-1-(\xi\lambda)^\quarter}\log\left( 1+\frac{2}{(\xi\lambda)^\quarter}\right)\label{longlower}\eea
which diverges logarithmically as $\xi\to 0$.

\subsection{Upper Bound as $\xi\to 0$}

\begin{figure}
\begin{center}
\includegraphics[width=14cm]{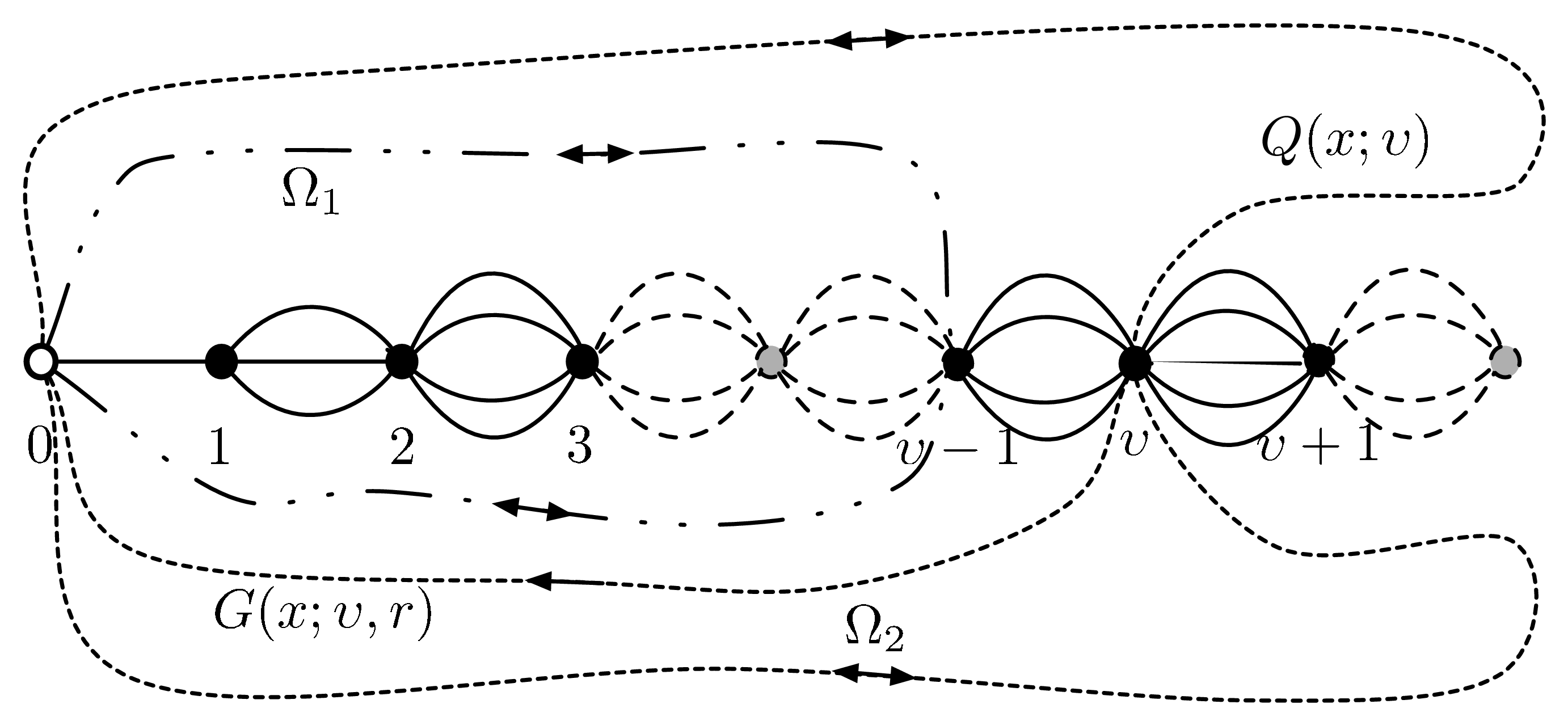}
\caption{Decomposition of a random walk into sets $\Omega _1$ and $\Omega _2$; $\Omega _1$ (two dots-one dash line) corresponds to walks that do not move beyond vertex $\upsilon -1$. $\Omega _2$ (dashed line) corresponds to walks that reach at least vertex $\upsilon$. Double arrows encode the fact that random walk can return to the point of origin and single arrow means that the walker cannot revisit its starting point.}
\label{omega}
\end{center}
\end{figure}

This proceeds by a fairly standard argument.
Let $p_t(r;n)$ be the probability that a walk starting at the root at time zero is at vertex $n$ at time $t$ and define
\beq Q(x;n)=\sum_t p_t(r;n)  (1-x)^{\half t}.\eeq
Then
\beq \sum_{n\in B(R)} Q(x;n) <\frac{2}{x}\eeq
so it follows that there is a vertex $v\le R$ such that
\beq Q(x;v)< \frac{2}{x R}.\eeq
Now consider the walks contributing to $Q(x)$ and split them into two sets; $\Omega_1$ consisting of those reaching no further than $v-1$,  and $\Omega_2$ consisting of those reaching at least as far as $v$ (see Figure \ref{omega}). Then we have 
\beq Q(x) = Q_{\Omega_1}(x)+ Q_{\Omega_2}(x).\eeq
$Q_{\Omega_2}(x)$ can be written
\beq Q_{\Omega_2}(x)=Q(x;v)\frac{L_{v-1}}{L_v+L_{v-1}} G(x;v,r)\eeq
where $G(x;v,r) $ generates walks which leave $v$ and return to $r$ without visiting $v$ again. $G(x;v,r)$ can be bounded by decomposing the walks as follows; leave $v$, go to $v-1$, do any  number of returns to $v-1$, leave $v-1$ for the last time, go to $v-2$ etc which gives
\beq G(x;v,r)=(1-x)^{-\half}P(x;v)\frac{L_{v-2}}{L_{v-1}}G(x;v-1,r)\eeq
where $P(x;v)$ is the first return generating function for walks that leave $v$ towards the root. Iterating gives (where the root is labelled 0)
\bea G(x;v,r)	&=&\prod_{k=v}^2(1-x)^{-\half}P(x;k)\frac{L_{k-2}}{L_{k-1}}\\
			&=& \frac{L_{0}}{L_{v-1}}(1-x)^{-\half(v-2)}\prod_{k=v}^2P(x;k).
\eea
One can then use the monotonicity Lemma to bound the $P(x;k)$ by reducing the multigraph to  $G_k^*=\{1,1,1,...,L_{k-1}(G),L_{k}(G),...\}$  which yields %
\beq
P(x;k) \leq P^*(x;k)=\frac{(1-x)\frac{ L_{k-1}}{1+L_{k-1}} }{1-\frac{1}{1+L_{k-1}} P_{k-1}(x)  } \leq (1-x)
\eeq 
where in the last inequality we used that $P_{k-1}(x)$, the first return generating function for walks on the line segment  of length $k-1$, is  bounded above by 1. Thus we have
\bea Q_{\Omega_2}(x)	&<& \frac{2}{x R}\frac{L_{0}}{L_v+L_{v-1}}          (1-x)^{ \half R} \\%
					&<& \frac{1}{x R}.    %
					\eea
$Q_{\Omega_1}(x)$ is bounded using \eqref{start} by
\beq Q_{\Omega_1}(x)=\eta_ {\Omega_1 0}<\eta_ {\Omega_1 v-2}+\sum_{n=0}^{v-3}\frac{1}{L_n}.\eeq
Now
\bea\eta_ {\Omega_1 v-2}	&=&\frac{1}{L_{v-2}}\;\frac{1}{1-\frac{L_{v-2}(1-x)}{L_{v-1}+L_{v-2}}} \\
			&=&\frac{1}{L_{v-2}} \;\frac{L_{v-1}+L_{v-2}}{L_{v-1}+xL_{v-2}}\\
			&<&\frac{1}{L_{v-2}}+\frac{1}{L_{v-1}}
\eea
so
\bea Q_{\Omega_1}(x)	&<&\sum_{n=0}^{v-1}\frac{1}{L_n}\\
					&<&  \sum_{n=0}^{R}\frac{1}{L_n} \eea
and altogether
\beq \label{upperbound}Q(x)<\frac{1}{x R}    %
+\sum_{n=0}^{R}\frac{1}{L_n}.\eeq
Taking expectation values
\bea \expect{Q(x)}{\infty U}	&<&\frac{1}{x R}   %
+\sum_{n=0}^{R}\frac{1}{1+f''(1)n/2}\\
			&<& \frac{1}{x R}    %
			+1+\frac{2}{f''(1)}\log( 1+f''(1)R/2 ).\eea
Finally let $R=bx^{-1}$   and set $b=a^\half\lambda^{-\half}$, $x=a\xi$ so 
\bea \tilde Q_1&=&\lim_{a\to 0}a^\half \avg{Q(a\xi)}\\
&<&{\lambda^\half}\left(1+ \log\left(   1+\frac{1}{\xi\lambda}  \right) \right).\eea
Together with \eqref{longlower} this establishes the second part of  \eqref{newthm1}.%

\section{Properties of spectral dimension in the transient case}\label{sec5}
 
Throughout this section we will assume that with measure $\mu=1$  there exist constants $c$ and $N_0$ such that
\beq B_{2N}(G)<c B_N(G),\quad N>N_0.\label{Ball:condition}\eeq
This rules out exponential growth for example.
Then a graph G has Hausdorff dimension $\dH$ if
\beq B_N(G)\sim N^{\dH}\eeq
and an ensemble $\R M=\{\C R,\mu\}$ has Hausdorff dimension $\dH$ if
\beq \expect{B_N(G)}{\mu}\sim N^{\dH}.\eeq
We will confine our considerations to the case $2\le \dH \le 4$ as being the regime of most physical interest. 

\subsection{Resistance and transience}

Here we note a lemma which describes the relationship between the resistance to infinity on $G$ and transience.

\begin{lemma}
The electrical resistance, assuming each edge has resistance 1, from $n$ to infinity is given by 
\beq \eta_n(0)=\sum_{k=n}^\infty \frac{1}{L_k}.\label{resistance}\eeq
If $\eta_n(0)$ is finite the graph is transient.
\end{lemma}
\begin{proof} This is a well known property of graphs (see for example \cite{LP:book}) but we give an explicit proof here as we will need \eqref{resistance} later. Clearly the right hand side of \eqref{resistance} is the resistance by the usual laws for combining resistors in parallel and series.  
 Setting $R= x^{-1}\abs{\log x}$ in \eqref{upperbound}  gives
\beq Q(x)<\frac{1}{\abs{\log x}}+\sum_{n=0}^\infty\frac{1}{L_n}-\sum_{n=R}^\infty \frac{1}{L_n}\eeq
and
\beq Q(0)=\eta_0(0)\le \sum_{n=0}^\infty\frac{1}{L_n}\label{etazero:upper}\eeq
which shows that if the rhs of this expression is finite the graph is definitely transient. Assuming this is the case, noting from \eqref{start}  
 that 
\beq \eta_1(0)=\eta_0(0)-1/L_0,\eeq
and proceeding by induction we see that 
\bea \eta_k(0)&=&\eta_0(0)-\sum_{n=0}^{k-1}\frac{1}{L_n}.\label{eta:reln}
\eea
Setting $k=\infty$ in  \eqref{eta:reln} gives 
 \beq  \eta_{0}(0)\ge \sum_{n=0}^\infty\frac{1}{L_n}\eeq
 which together with \eqref{eta:reln} and \eqref{etazero:upper} gives \eqref{resistance}. Finally it follows from \eqref{eta:def} that if $\eta_n(0)$ is finite, so is $Q_n(0)$ and therefore $G$ is transient. 
 \end{proof}
 Note that by  Jensen's inequality
  \beq \sum_{n=0}^\infty\frac{1}{L_n}=\lim_{N\to\infty}\sum_{n=0}^N\frac{1}{L_n}> \lim_{N\to\infty}\frac{N^2}{\sum_{n=0}^N L_n}\eeq
  so that only  if $\dH\ge 2$   can  the graph be transient.

The solvable case $L_n=(n+2)(n+1)$ is discussed in Appendix \ref{Solvable} as a simple illustration of all these properties.

\subsection{Universal bounds on $\eta_n'(x)$}

Differentiating the recursion \eqref{start} and iterating we obtain
\bea\fl{ \abs{\eta_0'(x)}=\abs{\eta_{N}'(x)} \prod_{k=0}^{N-1}\frac{(1-x)}{(1+xL_k\eta_{k+1}(x))^2}+}\nn\\{+\sum_{n=0}^{N-1} (L_n\eta_{n+1}(x)^2+\eta_{n+1}(x) )(1-x)^{-1}\prod_{k=0}^{n}\frac{(1-x)}{(1+xL_k\eta_{k+1}(x))^2}}\label{second}\\
\fl{\qquad\quad =\abs{\eta_{N}'(x)} \prod_{k=0}^{N-1}\frac{(1-xL_k\eta_{k}(x))^2}{1-x}+}\nn\\{+\sum_{n=0}^{N-1} (L_n\eta_{n+1}(x)^2+\eta_{n+1}(x) )(1-x)^{-1}\prod_{k=0}^{n}\frac{(1-xL_k\eta_{k}(x))^2}{1-x}}\label{first}\eea
so
\bea \fl{ \abs{\eta_0'(x)}
<\abs{\eta_{N}'(x)} (1-x)^{-N}e^{-2x\sum_{k=0}^{N-1} L_k\eta_{k}(x)} +}\nn\\{ +\sum_{n=0}^{N-1}( L_n\eta_{n+1 }(x)^2+\eta_{n+1}(x)) (1-x)^{-n-2}e^{-2x\sum_{k=0}^{n-1} L_k\eta_{k}(x)}} \label{etaprime:upper} \eea
and
\bea \fl {\abs{\eta_0'(x)}>\abs{\eta_{N}'(x)} (1-x)^N e^{-2x\sum_{k=0}^{N-1} L_k\eta_{k+1}(x)}+}\nn\\{+\sum_{n=0}^{N-1} (L_n\eta_{n+1}(x)^2+\eta_{n+1}(x) )(1-x)^n e^{-2x\sum_{k=0}^{n} L_k\eta_{k+1}(x)}}. \label{etaprime:lower} \eea
We see that the upper and lower bounds are essentially of the same form.  Defining
\bea F_N(x)=\sum_{k=0}^{N} L_k\eta_{k+1}(x), \label{FN:defn}\\
G_N(x)=\sum_{k=0}^{N} L_k\eta_{k+1}(x)^2+\eta_{k+1}(x)\label{GN:defn}
 \eea
we have
\begin{lemma}\label{lemma:etaprime}
For any  $G\in \C R$
\bea \abs{\eta_0'(x)}&>& c\,  G_{N^-(x)-1}(x)\label{new:bound}\eea
for $x< x_0<1$, where $c$ is a constant and $N^{-}(x)$ is the integer such that
\bea x F_{N^-(x)}(x)& >& 1\ge x F_{N^-(x)-1}(x).\label{FNminus:defn}\eea
\end{lemma}

\begin{proof} %
Setting $N=N^-(x)$ in   \eqref{EtaLower}, and using \eqref{FNminus:defn}
we have %
\bea \eta_n (x)	&>& e^{-1} \sum_{k=n}^{N^-(x)-1}\frac{1}{L_k}\label{eta:lower2}\eea
so that (using Jensen's inequality and the condition \eqref{Ball:condition}, see \eqref{A:result1})
\bea \frac{1}{x}>F_{N^-(x)-1}(x)	&>& e^{-1} \sum_{n=0}^{N^-(x)-1} L_n\sum_{k=n+1}^{N^-(x)-1}\frac{1}{L_k} > b^{-2} (N^-(x)-1)^2,\label{Nminus:inequalities}\eea
where $b$ is a constant $O(1)$, so that 
\beq \lceil b\, x^{-\half}\rceil >N^-(x).\label{Nminus:upper}\eeq
Lemma \ref{lemma:etaprime} then follows by setting $N=N^-(x)$ in \eqref{etaprime:lower} and using  \eqref{Nminus:upper}. For future use we note that because $\eta_k(x)<\eta_k(0)<\eta_0(0)$ we have
\beq \lfloor c\, x^{-1/\dH}\rfloor <N^-(x).\label{Nminus:lower}\eeq

\end{proof}

\subsection{Relationship between Hausdorff and spectral dimensions}

Our first result is 

\begin{theorem}  \label{ds leq dh} For any  graph $G\in \C R$ such that the Hausdorff dimension $\dH$ exists and is less than 4 then, if the spectral dimension exists, $d_s\le \dH$.
  \end{theorem}
\begin{proof}
The proof is by contradiction. 
First define the set of numbers
\beq \C X=\{x_k: 1=x_kF_{k-1}(x_k), k=1,2,3\ldots\}.\eeq
%
Applying Cauchy Schwarz inequality to \eqref{GN:defn} we have
\beq G_N(x)> \frac{F_N(x)^2}{B_N}\eeq
so  applying  Lemma \ref{lemma:etaprime} gives
\bea\abs{\eta_0'(x_k)}	&>&\frac {c}{x_k^2 B_k}.\label{etaprime:Xk}\eea
%
%
We assume that $\dH$ exists for $G$ so that
\bea\label{eqthm8} \abs{\eta_0'(x_k)}&>&\frac {c}{x_k^2 k^{\dH}}\Psi(k),\eea
where $ \Psi(k)$ denotes a generic logarithmically varying function of $k$. %
Using  \eqref{Nminus:upper} and \eqref{Nminus:lower} (the latter if necessary to bound the logarithmic part)
 the right hand side is bounded below  by 
$c x_k^{-2+\dH/2}\Psi(x_k)$. Now assume that $d_s$ exists in which case $\abs{\eta_0'(x_k)}< c' x^{-2+d_s/2}_k$; 
 however if $d_s>\dH$ there exist an infinite number of values $x\in\C X$ arbitrarily close to zero which contradict this. Therefore $d_s\le \dH$.
 
Note that for $\dH=4$ we get $\abs{\eta_0'(x_k)}>\Psi(x_k)$. If $\Psi(x)$ is logarithmically diverging as $x\to 0$ then again we can conclude that $d_s\le \dH$; the case when $ \eta_0'(0)$ is finite and $\dH=4$ is more subtle and we will not pursue it here.

\end{proof}

We can obtain more specific information about the spectral dimension by being more specific about the properties of the ensemble $\R M=\{\C R,\mu\}$. Define
\bea B_N^{(1)}=\sum_{k=0}^NL_k\sum_{n=k+1}^\infty \frac{1}{L_n},\\
\overline B_N^{(2)}=\sum_{k=0}^NL_k\sum_{m=k}^N \frac{1}{L_m}\sum_{n=k+1}^N \frac{1}{L_n}.	
\eea
Then we have the following
\begin{lemma}\label{lemma:dS:upper}
Given a graph $G\in\C R$ such that $\dH$ exists and
\bea B_N^{(1)}\sim N^{2+\gamma},\qquad
\overline B_N^{(2)}\sim N^{4-\dH+\delta},	\label{lemma:dS:upper:conditions}
\eea
then
the spectral dimension if it exists must satisfy
\bea d_s
&\le&\dH-\frac{2\delta-(4-\dH)\gamma}{2+\gamma}.\label{FullSpectralDim}
\eea
\end{lemma}
\begin{proof}
The proof uses Lemma \ref{lemma:etaprime} to show that
\beq  \abs{\eta_0'(x)} > \underline{c}\, x^{-2+\alpha/2}\abs{\log x}^{\underline c'}\eeq
where 
$\alpha$ is given by the right hand side of \eqref{FullSpectralDim}.   Firstly by combining \eqref{eta:lower2} and the definition of $G_N(x)$ we have
\beq G_{N^-(x)-1}(x)\ge e^{-2} \overline B_{N^-(x)-1}^{(2)},\eeq
while from the definition of $N^-(x)$ and the fact that $\eta_n(x)$ is a decreasing function of $x$ we get   \eqref{FNminus:defn} 
\bea \frac{1}{x}<F_{N^-(x)}(x)	&<& B_{N^-(x)}^{(1)}.\eea
Lemma \ref{lemma:dS:upper} follows by combining these two results with the conditions \eqref{lemma:dS:upper:conditions}
and Lemma \ref{lemma:etaprime}. 
Again, the special case $\dH=d_s=4$ is  more subtle because $\eta_0'(0)$ might be finite and we will not pursue it here.

\end{proof}

 Now define
\bea 
 B_N^{(2)}=\sum_{k=0}^NL_k\sum_{m=k}^\infty \frac{1}{L_m}\sum_{n=k+1}^\infty \frac{1}{L_n}\,.	\eea
Then
\begin{lemma}\label{lemma:dS:lower}
Given a graph $G\in\C R$ such that $\dH$ exists and
\bea \eta_N(0)\sim N^{2-\dH+\rho},\qquad%
 B_N^{(2)}\sim N^{4-\dH+\delta'},	\label{lemma:dS:lower:conditions}
\eea
or an ensemble $\R M=\{\C R,\mu\}$ such that 
\bea \expect{\eta_N(0)}{\mu}\sim N^{2-\dH+\rho},\qquad
 \expect{B_N^{(2)}}{\mu}\sim N^{4-\dH+\delta'},%
\eea
then the spectral dimension is bounded by
\bea d_s
&\geq&\dH- \frac{(4-\dH)\rho-(2-\dH)\delta'}{2+\delta'-\rho}\label{FullSpectralDim2}
\eea
provided that $\rho\leq \delta'+1$.
\end{lemma}
\begin{proof} Note that since $\eta_N(x)$ is a finite convex decreasing  function in $x=[0,1)$ 
\beq \abs{\eta_{N}'(x) } < \frac{\eta_{N}(0)}{x},\eeq
and $G_N(x)<G_N(0)=B_N^{(2)}$,
 combining this with \eqref{etaprime:upper} gives
\bea \abs{\eta_0'(x)}
&<& (1-x)^{-N}\left(\frac{\eta_{N}(0)}{x}+B_N^{(2)} \right). \label{etaprime:upper:1} \eea
Choosing $N=x^{-\frac{1}{2+\delta'-\rho}} $ gives
\beq \abs{\eta_0'(x)}<   {\overline c}\, x^{-\frac{4-\dH+\delta'}{2+\delta'-\rho}}\abs{\log x}^{\overline c'}\eeq
for $x<x_0$ and provided that $\rho\leq \delta'+1$. The result for $d_s$ follows. In the case of the ensemble average we simply take the expectation value in
\eqref{etaprime:upper:1} before proceeding as before. 
\end{proof}

There are a number of constraints on and relations between the quantities
$\rho$, $\dH$, $\gamma$, $\delta$ and $\delta'$ which are summarized by

\begin{lemma}\label{constraints}
For any graph $G\in\C R$ such that $\dH$ exists
\bea \rho\ge0,\quad\delta'\ge0,\quad \delta'\ge 2\gamma, \eea
and for any graph $G\in\C R$ such that $\dH$ and $\rho$ exist
\bea \gamma=\rho,\quad \delta'=2\rho,\quad \delta=2\rho.\eea
\end{lemma}
The proofs are elementary manipulations and outlined in Appendix \ref{Simple}.

The main result of this section is
\begin{theorem}\label{Theorem:dS-dH} For  any graph $G\in \C R$  such that $\dH<4$ and $\rho$ exist the spectral dimension is given by
\beq d_s=\frac{2\dH}{2+\rho}.\eeq
\end{theorem}
\begin{proof} The theorem follows from the upper and lower bounds in Lemmas \ref{lemma:dS:upper}
 and \ref{lemma:dS:lower}
 and using the relations between $\delta$, $\delta'$, $\gamma$ and $\rho$ in Lemma \ref{constraints}.

\end{proof}

It is an immediate corollary of Theorem \ref{Theorem:dS-dH}  that $\rho=0$ is a necessary and sufficient  condition for 
$\dS=\dH$.

\section{Scale dependent spectral dimension in the transient case} \label{sec6}

In this section we extend our results from Section \ref{sec4} regarding the scale dependent spectral dimension in the recurrent case to the transient case. In particular, we are interested in a situation, motivated by the numerical simulations of CDT \cite{Ambjorn:2005db}, where there is a scale dependent spectral dimension from $d_s^0=2$ at short distances while at long distances $d_s^\infty=4$. 

While two-dimensional causal triangulations are composed of  triangles, their four-dimensional analogues are composed of  four-simplices. As we explained in more detail in a recent letter \cite{letter} one can  define an injection from four-dimensional CDT to a multigraph by keeping only the time-like edges between subsequent slices and shrinking the sections of spatial three-simplices to a single vertex. The full partition function then induces a measure on the multigraph.

In the absence of an analytical expression for the partition function, numerical simulations such as \cite{Ambjorn:2005db,CDTnumerics,CDTnumerics2} give us insights into the properties of the measure. For example, there is compelling evidence that
\beq
\avg{B_N}_\mu\sim N^4
\eeq
to leading order and that 
\beq
\avg{L_N}_\mu\sim N^3
\eeq
to leading order. However, since the computer simulations are implemented for the space-like three-simplices and not for the time-like edges there are likely to be sub-leading terms, as well as a factor of inverse Newton's constant multiplying the leading term. Using the insights of Section \ref{sec4} we  make the following ansatz for the behaviour of $\avg{L_N}_\mu$ in analogy to the two-dimensional case\footnote{We use ``$\simeq$" to denote equality up to a multiplicative constant.}
\beq
\avg{L_N}_\mu\simeq\nu N^3 +N.\label{avg:LN}
\eeq
Recall that in the two-dimensional case we had $\avg{L_N}_\mu=f''(1) N + 1$ so here $\nu$ takes the role of $f''(1)$.  We consider no quadratic term here, since it would come with a rather unnatural scaling of Newton's constant. If $\nu$ is small in \eqref{avg:LN} then loosely speaking at very large distances the Hausdorff dimension is 4 while at short distances the linear term dominates and the volume growth appears to be two-dimensional.  As we will see in the following, this ansatz together with two technical bounds on the resistance will give us a model of a multigraph ensemble with transient walks which has a scale dependent spectral dimension varying from $d_s^\infty=4$ at long distances to $d_s^0=2$ at short distances. In particular, we make the following assumptions on the multigraph ensemble:

\begin{assume}\label{a1} 
The multigraph ensemble $\R{M}=\{\C{R},\mu\}$ satisfies
\newcounter{saveenum}
\begin{enumerate}
\item \beq \label{assumption_i}\avg{L_N}_\mu\simeq\nu N^{3-\epsilon} +N,\eeq
\setcounter{saveenum}{\value{enumi}}
\end{enumerate}
with $\epsilon>0$ being arbitrarily small and for $\mu$-almost all multigraphs there exists a $N_0>0$ such that for $N>N_0$
\begin{enumerate}\setcounter{enumi}{\value{saveenum}}
\item \beq \label{assumption_ii}\eta_N (0) \leq \frac{N}{\avg{L_N}_\mu}\Psi_+(\sqrt{\nu} N^{1-\epsilon/2}) 
\eeq
\item \beq \label{assumption_iii}L_N \leq \avg{L_N}_{\mu} \Psi (\sqrt{\nu} N^{1-\epsilon/2}) 
\eeq
\end{enumerate}
where $\Psi (x)$ and $\Psi_+(x)$ are functions which diverge and vary slowly 
at $x=0$ and $x=\infty$. %
\end{assume}

The introduction of the arbitrarily small constant  $\epsilon>0$ is for  technical reasons and for all practical purposes one can think of it as being zero. The bounds \eqref{assumption_ii} and \eqref{assumption_iii} are completely analogous to the concrete two-dimensional model studied in Section \ref{sec4}, where one has for example that, for almost all graphs,  $ L_N \leq \avg{L_N}_{\mu} \log (f''(1) N)$, for $N > N_0$. Assumption \eqref{assumption_i} implies $cN\avg{L_N}_\mu<\avg{B_N}_\mu<c' N\avg{L_N}_\mu$, where $c<c'$ are positive constants, and the multigraph ensemble $\R{M}=\{\mathcal{R},\mu\}$ is almost surely transient. 
 
The main result of this section is 

\begin{theorem}\label{VariableSpectralDimensionTransient}
A multigraph ensemble $\R{M}=\{\mathcal{R},\mu\}$ which satisfies Assumption \ref{a1} has $d_s^0=2$ at short distances while at long distances $d_s^\infty=4-\epsilon$ for $\epsilon>0$ arbitrarily small.
\end{theorem}

To prove the theorem we need the following lemma which we will prove in the following two  subsections
\begin{lemma} \label{lm:6.1}
For a multigraph ensemble $\R{M}=\{\C{R},\mu\}$ satisfying Assumption \ref{a1} 
\beq
\! \! \! \! \! \! \! \! \! c_-  \frac{1}{\nu b' x^{\epsilon/2} +x} < \avg{\abs{\eta_0'(x)}} _{\mu}< c_+ \frac{1}{\nu b' x^{\epsilon/2}+x}  \Psi_+^2\left( \sqrt{\frac{\nu}{x^{1-\epsilon/2}}} \right),
\eeq
where $c_-, c_+,b'$ are positive constants.
\end{lemma}
Defining the scaling limit as 
\beq
\abs{ \tilde Q'(\xi,G)} =\lim_{a\to 0} \left (\frac{a}{G} \right )\avg{\left |Q'\left(x=a\xi , \nu=\frac{a^{1-\epsilon/2}}{b'G}\right)\right |}_{\mu} 
 \eeq
 and using Lemma \ref{lm:6.1} gives
\bea
c_-\frac{1}{\xi^{\epsilon /2}+G\xi } < \abs{\tilde Q'(\xi,G)} < c_+ \frac{1}{\xi^{\epsilon/2} +G\xi }  \Psi_+^2\left( \sqrt{\frac{1}{G\xi^{1-\epsilon/2}}} \right).
\eea
For  short walks or equivalently $ \xi\gg G^{-1}$ we see that $\abs{\tilde Q'(\xi,G)}\sim\xi^{-1}$ giving $d_s^0=2$ while  for long walks or $ \xi\ll G^{-1}$ the $\abs{\tilde Q'(\xi,G)}\sim\xi^{-\epsilon/2}$  which leads to $d_s^\infty=4-\epsilon$ which completes the proof of the main theorem.

\subsection{Lower bound}
We now prove the lower bound of Lemma \ref{lm:6.1}. We begin with Lemma \ref{lemma:etaprime}  and proceed  as in the proof of Theorem \ref{ds leq dh} by applying the Cauchy Schwarz inequality to get
\bea \label{dimred lower bound1}
\abs{\eta '_0(x)} &>& c \frac{\left (\sum _{n=0} ^{N^{-}-1}L_n \eta_{n+1}(x) \right )^2}{\sum _{n=0}^{N^{-}-1}L_n} \\
                           &=& c \frac{\left (\sum _{n=0} ^{N^{-}}L_n \eta_{n+1}(x) - L_{N^-}\eta _{N^{-}+1}(x) \right )^2}{\sum _{n=0}^{N^{-}-1}L_n}.
\eea
We recall \eqref{FNminus:defn} and \eqref{Nminus:upper} to bound the sums in the numerator and denominator respectively. In addition we use the fact that $\eta_{N^{-}+1}(x) < \eta_{N^{-}}(x)<\eta_{N^{-}}(0)$ and assumptions \eqref{assumption_ii} and \eqref{assumption_iii} to get
\bea
 \abs{\eta '_0(x)} &>& c \frac{\left (\frac{1}{x} - N^{-}(x) \Psi \left ( \sqrt{\nu} (N^{-})^{1-\epsilon/2}\right ) \Psi_+\left ( \sqrt{\nu} (N^{-})^{1-\epsilon/2}\right ) \right )^2}{\sum _{n=0}^{N^{*}}L_n} \\
                           &>& c \frac{\left (1-x N^{*}(x) \Psi \left ( \sqrt{\nu} (N^{-})^{1-\epsilon/2}\right ) \Psi_+\left (\sqrt{\nu} (N^{-})^{1-\epsilon/2}\right)  \right )^2}{x^2\sum _{n=0}^{N^{*}}L_n} ,
\eea
where $N^* = \lceil bx^{-\half} \rceil$.  Since $\Psi(x)$, $\Psi_+(x)$ are slowly varying functions the second term in the numerator is sub-leading as $x\to0$.
Taking the expectation value and applying Jensen's inequality we find for $x<x_0<1$
\beq
 \avg{\abs{\eta '_0(x)}} _{\mu} > c \frac{1}{x^2\avg{\sum _{n=0}^{N^{*}}L_n}_{\mu}} 
\eeq
which, together with $\avg{B_N}_\mu<c' N\avg{L_N}_\mu$, implies the lower bound of Lemma \ref{lm:6.1} with $b'=b^{2-\epsilon}$.

\subsection{Upper bound}

To prove the upper bound we first note that from \eqref{etaprime:upper:1} 
\beq
\abs{\eta_0'(x)} < (1-x)^{-N} \left(   \frac{  \eta_{N}(0)  }{x}+B_N^{(2)} \right)
\eeq
for any $N$. Taking expectation values we now get
 \bea
\avg{\abs{\eta_0'(x)}}_\mu &<&(1-x)^{-N}  \left(   \frac{\avg{\eta_{N}(0)}_{\mu}}{x}+  \avg{B_N^{(2)}}_{\mu} \right)  \\
                                     &<& (1-x)^{-N}  \left (\frac{N}{x \avg{L_N}_{\mu}} + \frac{N^3}{\avg{L_N}_{\mu}} \right ) \Psi ^2_{+}(\sqrt{\nu}N^{1-\epsilon/2})
\eea
where we used the fact that $\avg{B_N^{(2)}}_{\mu} < \textrm{const} + c_3 \frac{N^3}{\avg{L_n}_{\mu}} \Psi ^2_{+}(\sqrt{\nu}N^{1-\epsilon/2})$; to prove this  proceed similarly to the upper bound of \eqref{B1:upper} using  $cN\avg{L_N}_\mu<\avg{B_N}_\mu<c' N\avg{L_N}_\mu$ together with \eqref{assumption_ii}. Choosing $N=\lceil b x^{-\half}\rceil$ gives the upper bound of Lemma \ref{lm:6.1}.

\section{Conclusions}


In this paper we have studied random walks on multigraph ensembles motivated by their close relationship to various causal quantum gravity models. In particular this approach is well suited to studying the spectral dimension and exploring its possible  scale dependence,   first introduced in the context of random combs in \cite{Atkin:2011ak},  in causal quantum gravity.
%

In the simplest model, discussed in Section \ref{sec4}, the measure on the multigraph is induced by the uniform measure on infinite causal triangulations or, equivalently, a critical Galton Watson process conditioned on non-extinction. This multigraph ensemble has Hausdorff dimension $\dH=2$. We show that by scaling the variance of the Galton Watson process to zero at the same time as one scales the walk length to infinity (cf. \eqref{scalingQ}) one obtains a continuum limit with a scale dependent spectral dimension which is $d_s^\infty=2$ at large scales and $d_s^0=1$ at small scales. Here $\lambda^{-\half}$ is related to the rescaled second moment $f''(1)$ of the branching process and $\lambda$ determines the scale separating the short and the long walk limit. Regarding the physical interpretation of this model two comments are in order:
\begin{enumerate}
\item In pure two-dimensional CDT there is no dependence on  Newton's constant due to the Gauss-Bonnet theorem. Hence, there is no length scale such as the Planck length in the model. This is  reflected in the fact that the uniform measure on infinite causal triangulations corresponds to a critical Galton Watson process with off-spring distribution $p_k=2^{-k-1}$ which has $f''(1)=2$ fixed. 
On the other hand as discussed in Section \ref{sec4}  the model with arbitrary $f''(1)$ can be thought of as describing CDT with a weight in the action coupling to the absolute value of the curvature \cite{Durhuus:2009sm} 
and $\sqrt{\lambda}$ viewed as renormalised two dimensional analogue of the gravitational constant $G^{(2)}$. In the vein of a recently introduced model of (2+1)-dimensional Ho\v rava-Lifshitz gravity \cite{Anderson:2011bj}, one might view this model as a toy model of (1+1)-dimensional Ho\v rava-Lifshitz gravity in the sense that it describes CDT with a higher curvature term.
\item Another point of interest is the dynamics of the model. The  model of random combs with scale dependent spectral dimension introduced in \cite{Atkin:2011ak} is a purely kinematic model proposed to show the existence of the scaling limit in a simplified context. On the other hand, the multigraph ensemble introduced in Section \ref{sec4} is directly related to CDT. It was shown in \cite{Sisko} that the rescaled  length process $l(t)=2a L_{[t/a]}/f''(1)$ of the multigraph  is described by the usual CDT Hamiltonian in the continuum limit
\beq
\hat{H}=-2\frac{\partial}{\partial l}-l \frac{\partial^2}{\partial l^2}+2 \mu l,
\eeq
where $\mu$ is the cosmological constant.
\end{enumerate}


In Section \ref{sec5} general properties of multigraph ensembles with transient walks are discussed, restricting ourselves to the physically interesting regime with $2\leq \dH \leq 4$. The main results are that for any multigraph $G$ such that $\dH$ and $d_s$ exist one has
\beq
d_s\leq\dH.
\eeq
If in addition the resistance exponent $\rho$  exists, then
\beq
d_s=\frac{2\dH}{2+\rho}.
\eeq
This implies that  $\rho=0$, which is a purely geometrical condition on the distribution of edges,  is a necessary and sufficient condition for $\dS=\dH$. It is interesting to notice in this context how multigraphs with $\rho=0$ attain the upper bound in \eqref{relationdsdh}.
%


In Section \ref{sec6} we propose a model of a multigraph ensemble with scale dependent spectral dimension in the transient regime. In particular, we assume that the measure $\mu$  satisfies
\beq
\avg{L_N}_\mu\simeq\nu N^3 +N,
\eeq
in addition to two more technical properties stated in assumptions \eqref{assumption_ii} and \eqref{assumption_iii}. It is then shown that this multigraph ensemble has a scale dependent spectral dimension with $d_s^0=2$ at short scales while at long scales $d_s^\infty=4$. Regarding the physical interpretation of this model we can make the following comments:
\begin{enumerate}
\item The scale separating the regime of $d_s^0=2$ and $d_s^\infty=4$ is set by the rescaled $G=a/\nu$. Viewing the multigraph ensemble as a model of four-dimensional causal quantum gravity one can interpret $G$ as the renormalised Newton's constant. While $G$ sets a scale on the duration of the walk, it is $\sqrt{G}$ that corresponds to the extent of the walk distance on the graph which can be identified with the Planck length $l_P$. 
\item In  \cite{letter} we propose the model discussed in Section \ref{sec6} as a concrete model of four-dimensional CDT. In particular, the model gives some analytical understanding of the numerical results presented in \cite{Ambjorn:2005db}
and is able to explain in a natural way the expression conjectured there for
 the continuum return probability density $\tilde{P}(\sigma)$ as a function of continuous diffusion time $\sigma$
\bea \label{pav-scale}
\tilde{P}(\sigma) \sim   \frac{2 G^2}{\sigma^2}  \frac{1}{ 1 + 2G / \sigma}.
\eea
\end{enumerate}


As a final comment we would note that the framework developed in this article is quite universal;  it can   be used to describe the phenomenon of a scale dependent spectral dimension in any theory with a fundamental discreteness scale and is equally applicable, for example, to three-dimensional causal quantum gravity. We leave further applications to future work.

\ack{
GG would like to acknowledge the support of  the A.G. Leventis Foundation and the A.S. Onassis Public Benefit Foundation grant F-ZG 097/ 2010-2011. This work is supported by EPSRC grant EP/I01263X/1 and STFC grant ST/G000492/1. SZ would like to acknowledge support of the STFC, as well as of a EPRSC visiting grant and Visiting Scholarship at Corpus Christi College, Oxford University during which this work was started.}

\appendix
\renewcommand\thesection{\Alph{section}}

\section{Basic solvable examples}\label{Solvable}

We  give two exactly solvable examples to illustrate some of the features derived in this paper.

\subsection{Recurrent case: spectral dimension of the half line}

First  we determine the spectral dimension of the half line through the generating function formalism (a slightly more detailed discussion can for example be found in \cite{Durhuus:2005fq}). Note that the half line is a special multigraph $\{L_k=1,k=0,1,...\}$ and as we see in Lemma \ref{monotonicity} it plays an important role in providing certain upper bounds for the return probability on any multigraph.

Since the random walk on the half line has to leave the root with probability one and otherwise can move to either neighbour  with probability $1/2$  the generating function for  the first return probability satisfies
\beq \label{Phalfline}
P_\infty(x) =\frac{1-x}{2-P_\infty(x)}.
\eeq
From this we get that 
\beq
P_\infty(x) =1-\sqrt{x}
\eeq
and thus 
\beq \label{Qhalfline}
Q_\infty(x) =\frac{1}{\sqrt{x}}.
\eeq
This shows that $Q_\infty(x)$ diverges as $x\!\to\! 0$ 
and that the spectral dimension is $d_s=1$. From the multigraph point of view, one also has the trivial result that $B(N)=\sum^N_{k=0} L_k=N+1$ and thus that $d_s=\dH$.

If  instead of the half line we consider a line segment of length $L$, then \eqref{Phalfline} becomes
\beq \label{PhalflineL}
P_L(x) =\frac{1-x}{2-P_{L-1}(x)}
\eeq
for $L\geq 1$ and $P_0(x)=1$. This relation can be iterated to give \cite{Durhuus:2005fq}
\beq \label{Psegment}
P_L(x) =1- \sqrt{x} \frac{(1+\sqrt{x})^L - (1-\sqrt{x})^L}{(1+\sqrt{x})^L + (1-\sqrt{x})^L}.
\eeq
  Similarly the contribution to the first return probabilitty on the full half line from walks that do not extend beyond $L$ is given by
\beq \label{PhalflineLf}
R_L(x) =1- \sqrt{x} \frac{(1+\sqrt{x})^L +(1-\sqrt{x})^L}{(1+\sqrt{x})^L - (1-\sqrt{x})^L}.
\eeq
\subsection{Non-recurrent case: spectral dimension of a multigraph with $L_k\sim k^2$}

As an explicit example of a non-recurrent graph we consider a (fixed) multigraph $M=\{L_k,k=0,1,...\}$ with $L_k =(k +1)(k+2)$. This particular multigraph is rather special as we will see in the following. Note that the probability for a random walker at vertex $k+1$ of $M$ to go forward is 
\bea
p_{k+1} = \frac{L_{k+1}}{L_{k} +L_{k+1}}= \frac{k+3}{2(k+2)}
\eea
and the probability of returning from $k+2$ to $k+1$,
\bea
q_{k+1} =1- p_{k+2} = \frac{L_{k+1}}{L_{k+1} +L_{k+2}}= \frac{k+2}{2(k+3)}.
\eea
We observe that 
 \bea
p_k q_k = p_k (1-p_{k+1}) =\frac{k+2}{2(k+1)}\frac{k+1}{2(k+2)} = \frac{1}{4} 
 \eea
as was the case for the half line. Hence we can relate the first return generating function for a random walker on the multigraph $M_k$  to that on half line by just compensating for the last step of the random walk which on the half line would occur with probability $1/2$ while on $M_k$ it occurs with probability $q_k$ leading to
 \beq
P_{M_k}(x) = 2 q_k P_{\infty}(x).
 \eeq
It follows that
\bea
\label{etak2}
\eta _{M_k}(x) \equiv \frac{Q_{M_k(x)}}{L_k}= \frac{1}{L_k}\frac{1}{1-P_{M_k}(x) } =\frac{1}{(k+1)(1+(k+1)\sqrt{x})}.
\eea

We note that $\eta_k(0)=1/(k+1)$ is finite which shows that the random walk is non-recurrent. The first derivative is
\beq
-\eta _{M_0}'(x) = \frac{1}{2(1+\sqrt{x})^2 \sqrt{x}}\sim x^{-1/2} \quad \textrm{as}\quad x\to 0
\eeq
and hence the spectral dimension is $d_s=3$ while  $B(N)=\sum^N_{k=0} L_k\sim N^3$ and thus $d_s=\dH$. It is straightforward to check that the resistance exponent $\rho=0$.

\section{Simple Results for Lemma \ref{constraints} }\label{Simple}
Here we outline the proofs of Lemma \ref{constraints} assuming unless otherwise stated that $\dH$ exists and  $N>N_0$.
To prove that $\rho\ge 0$ note that 
\bea  \eta_{N}(0)	=\sum_{n=N}^{\infty} \frac{1}{L_n}
				>\sum_{n=N}^{2N} \frac{1}{L_n}
				>\frac{(2N-N)^2}{\sum_{n=N}^{2N} L_n}
				\sim {\rm const}\,N^{2-\dH},   \label{etalwrbnd} \label{A:result1}\eea
where we have used Jensen's inequality. To prove that $\delta'\ge 0$ note that 
\bea  B_N^{(2)}&=&\left(\sum_{k=0}^{N_0}+\sum_{k>N_0}^{N}\right) L_k\sum_{m=k}^\infty \frac{1}{L_m}\sum_{n=k+1}^\infty \frac{1}{L_n}\nonumber\\	
 &=& \textrm{const} + \sum_{k>N_0}^{N} L_k\sum_{m=k}^\infty \frac{1}{L_m}\sum_{n=k+1}^\infty \frac{1}{L_n} \eea
 and then use \eqref{etalwrbnd}; similar arguments show that $\gamma\ge0$ and $\delta\ge0$.
 Using Cauchy Schwarz inequality
\beq
\left(\sum_{k=0}^{N-1} L_k\eta_{k+1}(0) \right)^2 < N^{\dH}\sum_{k=0}^{N-1} L_k\eta_{k+1}(0) ^2\eeq
which gives $\delta'\ge 2\gamma$.

To establish the relation between the other exponents and $\rho$, assuming it exists, we need the inequalities
%
that for any graph $G\in \C R$
\bea  \! \! \! \! \!  &(i)&\! \! \quad \eta_{N+1}(0) B_N< B_N^{(1)} < B_{N_1}^{(1)}+ \! \! \! \sum_{r=1}^{\lceil \log_2\frac{N}{N_0}\rceil }  \! \!  \eta_{\lceil N/2^r \rceil }(0)(B_{\lceil N/2^{r-1}\rceil}-B_{\lceil N/2^r \rceil}) \label{B1:upper}    \\
\! \! \! \! \! &(ii)&\quad \frac{\left(B_N^{(1)}\right)^2}{ B_N}<B_N^{(2)} < B_{N_1}^{(2)}+ \! \! \! \sum_{r=1}^{\lceil \log_2\frac{N}{N_0}\rceil }   \eta_{\lceil N/2^r \rceil}(0)(B_{\lceil N/2^{r-1}\rceil }^{(1)} -B_{\lceil  N/2^r \rceil}^{(1)})  \label{B2:upper} \eea
for $N_0\le N_1\le 2 N_0 < N$.
The proofs exploit the fact that $\eta_k(0)$ is a decreasing sequence. For example to prove (i)  we have the lower bound
\bea B_N^{(1)}>\eta_{N+1}(0)\sum_{k=0}^{N}L_k=\eta_{N+1}(0) B_N.\eea
An upper bound is given by
%
\bea B_N^{(1)}	&=&\sum_{k=0}^{\lceil N/2 \rceil} L_k\eta_k(0)+\sum_{k=\lceil N/2 \rceil+1}^N  L_k\eta_k(0)\nn \\
			&<&B_{\lceil N/2 \rceil}^{(1)}+\eta_{\lceil N/2 \rceil}(0)(B_N-B_{\lceil N/2\rceil})\eea
and iterating to get \eqref{B1:upper}.  The proof of (ii) proceeds analogously.
Assuming that  $\rho$ exists then it follows from \eqref{B1:upper} and \eqref{B2:upper} respectively that
\beq \gamma=\rho,\quad \delta'=2\rho.\eeq
By definition $\delta\le \delta'$ and noting
 that 
\bea   \overline B_N^{(2)} 		&>& \sum_{k=0}^{\lfloor N/2 \rfloor }L_k\left(\sum_{n>\lfloor N/2\rfloor}^N \frac{1}{L_n}\right)^2\\
						&=&c \left(\frac{N}{2}\right)^\dH\left (\eta_{\lfloor N/2\rfloor}(0)-\eta_N(0)\right)^2\\
						&=& c' N^{4-\dH+2\rho}\eea
we also have  $\delta\ge2\rho$ so conclude that $\delta=2\rho$.

\end{document}